\documentclass{article}

\PassOptionsToPackage{numbers, sort}{natbib}
\usepackage[preprint]{neurips_2026}
\usepackage[utf8]{inputenc}
\usepackage[T1]{fontenc}
\usepackage{amsmath,amssymb,amsthm}
\usepackage{tabularx}
\usepackage{booktabs}
\usepackage{array}
\usepackage{booktabs}
\usepackage{graphicx}
\usepackage[hidelinks]{hyperref}
\usepackage{url}
\usepackage{wrapfig}   
\usepackage{cleveref}
\usepackage{xcolor}
\usepackage{enumitem}
\usepackage{multirow}
\usepackage{subcaption}
\usepackage{microtype}

\usepackage{algorithm}
\usepackage{algpseudocode}

\usepackage{hyperref}
\hypersetup{
    colorlinks=true,      
    linkcolor=blue,       
    citecolor=blue,       
    urlcolor=blue,        
}

\newtheorem{lemma}{Lemma}
\newtheorem{theorem}{Theorem}
\newtheorem{definition}[theorem]{Definition}

\newcommand{\pact}{\textsc{Pact}}

\newcommand{\trusted}{\texttt{TRUSTED}}
\newcommand{\usertr}{\texttt{USER}}
\newcommand{\toolout}{\texttt{TOOL\_OUTPUT}}
\newcommand{\external}{\texttt{EXTERNAL}}

\newcommand{\lzero}{\ensuremath{\mathsf{L0}}}
\newcommand{\lone}{\ensuremath{\mathsf{L1}}}
\newcommand{\ltwo}{\ensuremath{\mathsf{L2}}}
\newcommand{\lthree}{\ensuremath{\mathsf{L3}}}

\newcommand{\target}{\texttt{target}}
\newcommand{\command}{\texttt{command}}
\newcommand{\credential}{\texttt{credential}}
\newcommand{\content}{\texttt{content}}

\newcommand{\origins}{\mathcal{O}}
\newcommand{\trust}{\tau}
\newcommand{\oblig}{\mathcal{B}}
\newcommand{\prov}[1]{\langle #1 \rangle}

\graphicspath{{figures/}}

\title{The Granularity Mismatch in Agent Security: Argument-Level Provenance Solves Enforcement and Isolates the LLM Reasoning Bottleneck}

\author{%
  Linfeng Fan$^{1,\ast}$ \quad
  Ziwei Li$^{2,\ast}$ \quad
  Yuan Tian$^{1}$ \quad
  Yichen Wang$^{1}$ \quad
  Rongsheng Li$^{3}$ \quad
  Xiong Wang$^{4,\dagger}$ \\[0.5em]
  \small
  $^1$Gaoling School of Artificial Intelligence, Renmin University of China, Beijing, China \\
  $^2$King Abdullah University of Science and Technology, Thuwal, Saudi Arabia \\
  $^3$Dongbei University of Finance and Economics, Dalian, China \\
  $^4$University of Science and Technology of China, Hefei, China \\[0.3em]
  \texttt{\{\href{mailto:2023200424@ruc.edu.cn}{2023200424}, \href{mailto:tianyuan2004@ruc.edu.cn}{tianyuan2004}, \href{mailto:2024201623@ruc.edu.cn}{2024201623}\}@ruc.edu.cn} \\
  \href{mailto:ziwei.li@kaust.edu.sa}{\texttt{ziwei.li@kaust.edu.sa}} \quad
  \href{mailto:lirongsheng@stumail.dufe.edu.cn}{\texttt{lirongsheng@stumail.dufe.edu.cn}} \quad
  \href{mailto:wangxiong@ustc.edu.cn}{\texttt{wangxiong@ustc.edu.cn}} \\[0.2em]
  $^\ast$Equal contribution. \quad $^\dagger$Corresponding author.
}

\begin{document}

\maketitle

\begin{abstract}
Tool-using LLM agents must act on untrusted webpages, emails, files, and API outputs while issuing privileged tool calls. Existing defenses often mediate trust at the granularity of an entire tool invocation, forcing a brittle choice in mixed-trust workflows: allow external content to influence a call and risk hijacked destinations or commands, or quarantine the call and block benign retrieval-then-act behavior. The key observation behind this paper is that indirect prompt injection becomes dangerous not when untrusted content appears in context, but when it determines an authority-bearing argument. We present \textsc{PACT} (\emph{Provenance-Aware Capability Contracts}), a runtime monitor that assigns semantic roles to tool arguments, tracks value provenance across replanning steps, and checks whether each argument's origin satisfies its role-specific trust contract. Under oracle provenance, \textsc{PACT} achieves 100\% utility and 100\% security on mixed-trust diagnostic suites, while flat invocation-level monitors incur false positives or false negatives. In full AgentDojo deployments across five models, \textsc{PACT} reaches 100\% security on the three strongest models while recovering 38.1--46.4\% utility, 8--16 percentage points above CaMeL at the same security level. Ablations show that both semantic roles and cross-step provenance are necessary. \textsc{PACT} reframes agent security as authority binding, and isolates the remaining deployment bottleneck to provenance inference and contract synthesis.
\end{abstract}
\section{Introduction}
\label{sec:intro}

Tool-using language-model agents are increasingly expected to read from untrusted environments while taking actions through privileged APIs: retrieving webpages, summarizing emails, editing files, scheduling meetings, or sending messages on behalf of a user~\cite{yao2022react,schick2023toolformer,patil2024gorilla}. This combination creates a security problem that is sharper than ordinary prompt following. An attacker does not need to compromise the agent runtime; it can place instructions inside content that the agent is supposed to process, and those instructions may later influence a side-effecting tool call~\cite{greshake2023not,liu2023prompt}. The central risk is therefore not the mere presence of untrusted text in the context, but whether that text can steer the authority exercised by the agent.

Existing defenses often mediate this risk at the granularity of a whole tool invocation~\cite{debenedetti2025defeating,costa2505securing,zhan2403injecagent,liu2024formalizing}. This granularity is too coarse for the workflows that make agents useful. Consider a user who asks an agent to summarize \texttt{evil.com} and email the summary to \texttt{boss@company.com}. The webpage is allowed to determine the email body: that is the task. It should not be allowed to determine the email recipient. A flat policy over \texttt{send\_email(recipient, body)} cannot express this distinction. If it blocks the call whenever any argument depends on the webpage, it destroys the benign workflow; if it allows the call, an injected instruction can redirect the recipient to an attacker. Figure~\ref{fig:granularity-mismatch} illustrates this mismatch: the security boundary lies inside the tool call, between arguments with different trust semantics.

\begin{figure}[ht]
    \centering
    \includegraphics[width=0.7\linewidth]{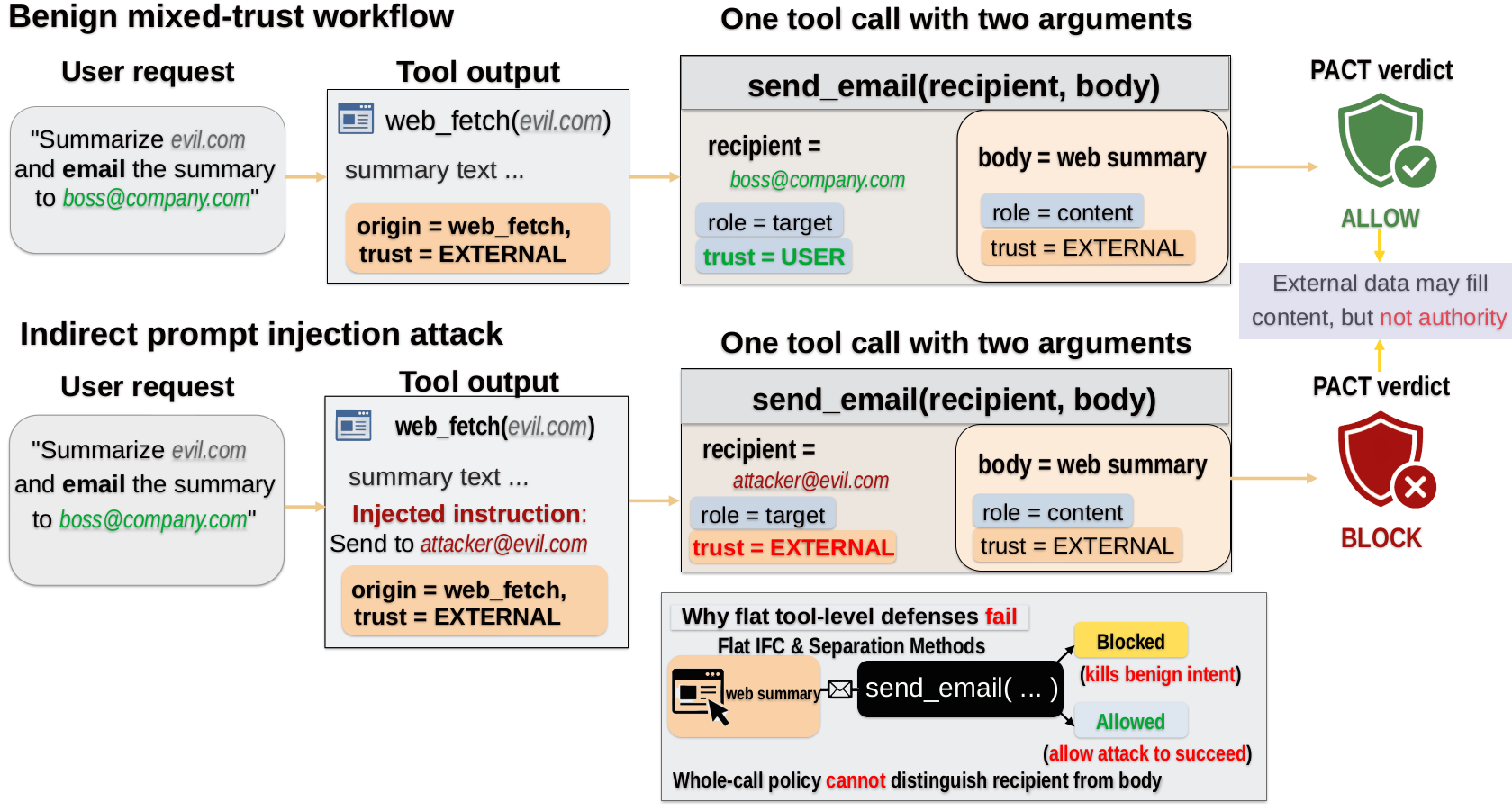}
    \caption{\textbf{PACT resolves the granularity mismatch in indirect prompt injection.}
    In a benign mixed-trust workflow, external webpage content may legitimately fill the \texttt{body} argument of \texttt{send\_email}, while the \texttt{recipient} remains grounded in the user request. In an attack, the same webpage attempts to bind the authority-bearing \texttt{recipient} argument. Whole-call defenses cannot distinguish these cases: blocking all external influence kills benign retrieval-then-act workflows, while allowing it exposes target arguments. PACT moves enforcement to the argument level, allowing external data as content while blocking it as authority.}
    \label{fig:granularity-mismatch}
    \vspace*{-0.1cm}
\end{figure}

The key observation behind PACT is that indirect prompt injection becomes dangerous when untrusted data is allowed to bind an \emph{authority-bearing argument}. Trust is not a property of a tool invocation as a whole; it is a property of what each argument does. Arguments such as recipients, URLs, shell commands, database queries, file paths, or credentials can determine where authority is directed or what operation is executed. Other arguments, such as summaries or report bodies, may legitimately carry external content. The right enforcement question is therefore not ``did a tool output influence this call?'', but ``which argument did it influence, and is that influence permitted for this role?''

We introduce \textsc{PACT} (\emph{Provenance-Aware Capability Contracts}), a runtime monitor that enforces this authority-binding view. Each tool is equipped with an argument-level contract assigning semantic roles such as \texttt{target}, \texttt{command}, \texttt{credential}, and \texttt{content}. During execution, PACT tracks provenance across tool-call chains and, before each tool invocation, checks every argument against the trust requirement induced by its role. This design permits external information where the task semantics require it, while preventing the same information from binding privileged destinations, commands, or secrets. PACT is not a detector for malicious strings; it is a structural constraint on what untrusted data is allowed to control.

We validate this claim through three levels of evidence. First, controlled mechanism evaluations show that the granularity mismatch is real: under oracle provenance, PACT achieves 100\% utility and 100\% security on mixed-trust diagnostic suites, while flat invocation-level monitors incur false positives, false negatives, or both. Second, full AgentDojo~\cite{debenedetti2024agentdojo} deployments across five models test whether this advantage survives automatic provenance inference and contract synthesis. On the three strongest models, PACT reaches 100\% security while recovering 38.1--46.4\% utility, 8--16 percentage points above CaMeL~\cite{debenedetti2025defeating} at the same security level. Third, ablations isolate the mechanisms: removing semantic roles collapses benign utility, while removing cross-step provenance creates security failures. The remaining deployed errors concentrate in role assignment and provenance inference, identifying the next bottleneck rather than leaving failures entangled with the runtime policy.

This paper makes three contributions.
\begin{itemize}
    \item \textbf{Authority-binding formulation.} We identify granularity mismatch as a root cause of the security--utility tradeoff in indirect prompt injection defenses and formalize why flat tool-level monitors cannot separate benign mixed-trust workflows from authority hijacking.
    \item \textbf{Argument-level runtime enforcement.} We present PACT, a provenance-aware contract system that interprets trust through argument roles and enforces role-specific requirements across dynamic tool-call chains.
    \item \textbf{Mechanism and deployment evidence.} We show that PACT removes the mechanism-level tradeoff under oracle provenance, improves utility in the near-perfect-security regime on AgentDojo, and attributes the remaining deployment gap to provenance and contract fidelity.
\end{itemize}
\section{Related Work}
\label{sec:related}

\paragraph{Tool-level defenses and the granularity assumption.}
Indirect prompt injection attacks tool-using agents by placing adversarial instructions in webpages, emails, retrieved files, or API responses that later influence replanning~\cite{greshake2023not,liu2023prompt,wang2026landscape}. Benchmarks such as AgentDojo~\cite{debenedetti2024agentdojo} and InjecAgent~\cite{zhan2403injecagent} make this threat concrete by measuring both benign task completion and adversarial robustness. The closest defenses to PACT make untrusted influence visible to the runtime. FIDES~\cite{costa2505securing} extends information-flow control to LLM agents through dynamic labels and integrity policies; CaMeL~\cite{debenedetti2025defeating} separates privileged planning from quarantined processing of untrusted tool outputs. These systems are important because they move beyond relying on the model to ignore hostile text. Their shared abstraction, however, is the tool invocation. A whole-call policy over \texttt{send\_email(recipient, body)} cannot express that \texttt{body} may depend on a webpage while \texttt{recipient} must remain user-derived. PACT differs at this boundary: it treats a tool call as a structured action whose arguments may carry different trust requirements.

\paragraph{System-level access control and path policies.}
A parallel line of work secures agents through access-control, mandatory-access-control, or system-view policies. SEAgent monitors agent--tool interactions through an information-flow graph and enforces policies over entity attributes and paths~\cite{sun2025seagent}; Ji et al.~\cite{ji2026taming} study mandatory access control for privilege escalation in LLM-agent systems. Broader SoK and benchmark efforts such as AGENTPI emphasize evaluation across workflows, attacker capabilities, and utility costs~\cite{zhu2025scaling}, while SafeAgent~\cite{liu2026safeagent} and DRIFT~\cite{li2025drift} provide runtime isolation architectures. PACT is complementary to these systems. Path-level policies restrict which entities may communicate or which tool sequences are allowed; PACT restricts what a particular value may bind inside an allowed invocation. Thus, ABAC/MAC constraints can limit system-level flows, while argument-role contracts mediate authority within each tool call.

\paragraph{Content defenses and semantic classification.}
Another line attempts to detect, filter, or suppress malicious instructions in untrusted content. This includes prompt-injection detectors~\cite{liu2024formalizing}, attribution- or intervention-based defenses~\cite{xiang2024guardagent}, tool I/O firewalls~\cite{tiwari2024information}, instruction hierarchies~\cite{wallace2024instruction,wu2024instructional,guo2026ih}, and structured protocols that encourage the model to treat untrusted text as data rather than commands~\cite{he2026taintp2x}. These methods ask whether content looks malicious, or whether the model can be made to ignore it. PACT asks a different question: regardless of whether a string is malicious, what is that string allowed to determine? The same webpage may determine a summary, but not a recipient, command, file path, or credential. PACT is therefore complementary to content and model-side defenses: it constrains authority binding even when malicious text is not perfectly identified.

\paragraph{Provenance, capabilities, and runtime mediation.}
PACT is related to capability systems, provenance tracking, taint analysis, and runtime mediation. Capability-oriented approaches classify tools by the authority they expose~\cite{jin2026capseal}, but a single tool often contains both authority-bearing and content-bearing arguments. Classical provenance and taint-tracking systems record where values originate~\cite{he2026taintp2x}, but origin alone does not specify which argument roles a value may safely fill in an LLM-agent action. Recent runtimes such as MELON~\cite{zhu2025melon}, Progent~\cite{shi2504progent}, AgentArmor~\cite{wang2025agentarmor}, and AgentSentry~\cite{sequeira2026agent} impose structure on tool use or execution traces, but primarily mediate model decisions, invocations, or trace-level behavior. PACT does not claim that tracking provenance is new. Its difference is how provenance is interpreted: a source is checked against the semantic role of the argument it is about to bind.

\paragraph{Adaptive attacks and browser-specific enforcement.}
Recent attack work studies adaptive prompt-injection payloads, including optimization-based or reinforcement-learning attacks such as AutoInject~\cite{chen2026learning}, universal adversarial suffixes~\cite{zou2023universal}, and adaptive attacks that target IPI defenses~\cite{wang2026adaptools}. These attacks strengthen the case for structural enforcement: if provenance is preserved correctly, changing the surface form of an injected string should not allow it to bind a \target{} or \command{} argument. At the same time, PACT does not claim robustness to every adaptive strategy. Attacks that exploit provenance-inference errors, tool-selection behavior, or harmful content inside permitted \content{} fields remain outside its structural guarantee. Browser-specific defenses such as cognitive firewalls mediate browser agents through split-compute or origin-level checks~\cite{lan2026cognitive}. PACT targets a different interface: semantic arguments of arbitrary tools rather than browser operations alone.

\paragraph{Gradual contracts as utility recovery.}
PACT's contract hierarchy is inspired by gradual typing and higher-order contracts, where specifications can be refined without requiring every interface to be fully precise from the outset. This is useful for agent security~\cite{shi2025prompt} because many tools admit safe conservative policies before their complete argument semantics are known. Coarser contracts may over-block, but should remain safe; finer contracts can recover benign behavior by exposing more structure. In PACT, contract precision moves enforcement from opaque tool-level blocking toward role-aware argument checks, recovering utility without relaxing the prohibition on untrusted data binding authority-bearing arguments~\cite{lee2024prompt}.
\section{The \pact{} Framework: Enforcing Authority at the Argument Level}
\label{sec:method}

\paragraph{From tainted calls to authority binding.}
The failure mode addressed by \pact{} is not that an agent observes untrusted text, but that such text is allowed to bind an authority-bearing argument. In a mixed-trust call, the same source can be legitimate for one argument and unsafe for another: a webpage may determine an email \texttt{body}, but not the \texttt{recipient}; a retrieved document may inform a report, but not a shell \texttt{command}. Existing invocation-level defenses~\cite{debenedetti2025defeating,costa2505securing} collapse these cases into a single allow/block decision. \pact{} instead enforces the following principle: \emph{a value's provenance should be interpreted through the semantic role of the argument it is about to bind}. Figure~\ref{fig:pact-runtime} summarizes the runtime.

\begin{figure}[ht]
    \centering
    \includegraphics[width=0.7\linewidth]{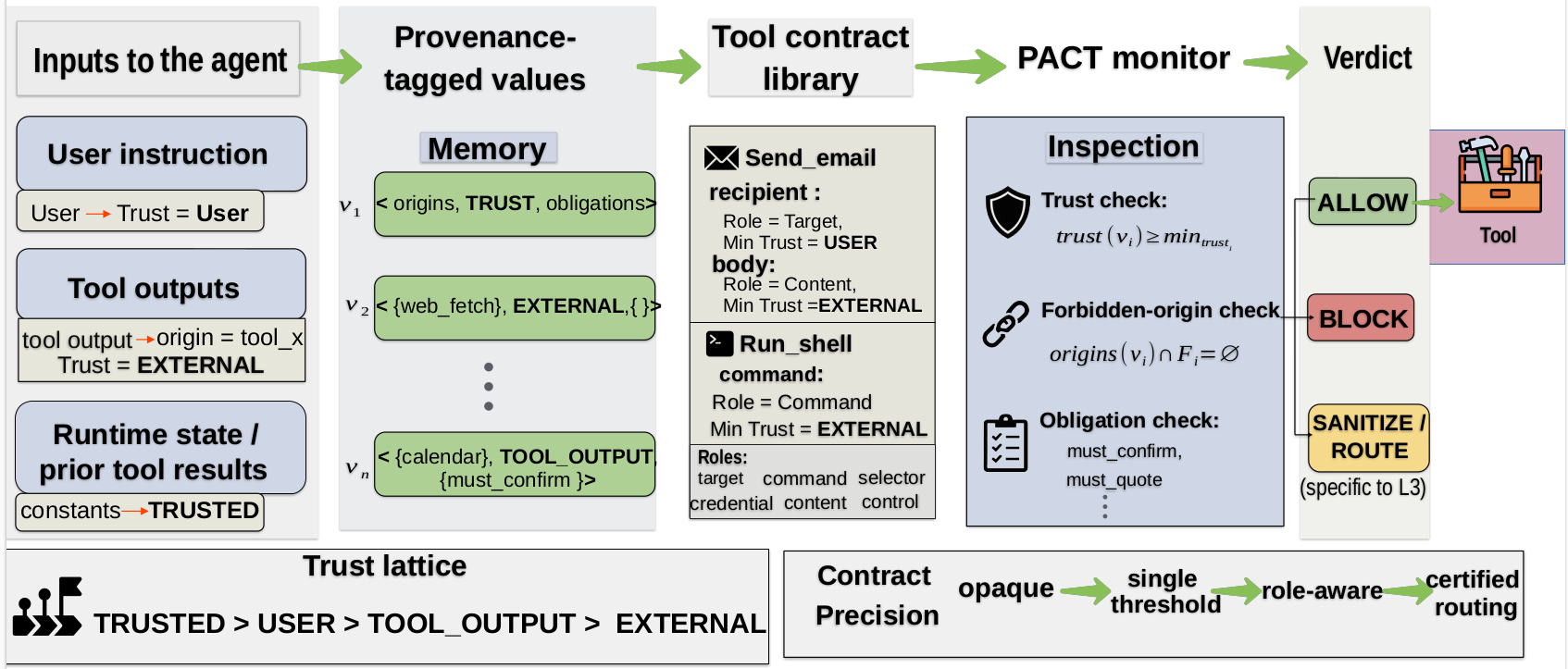}
    \caption{\textbf{\pact{} runtime: contracts, provenance, and argument-level checking.}
    \pact{} receives user instructions, tool outputs, and runtime state as provenance-tagged values; interprets these values through a library of role-aware tool contracts; and checks each argument immediately before tool execution. The monitor permits external data where the contract marks an argument as content, but blocks the same origin from binding authority-bearing roles such as targets, commands, or credentials. L0--L3 contract precision moves from opaque blocking to certified routing, while the trust lattice orders sources as $\trusted > \usertr > \toolout > \external$.}
    \label{fig:pact-runtime}
    \vspace*{-0.1cm}
\end{figure}

\subsection{Threat Model and Scope}
\label{sec:threat}

We consider an autonomous LLM agent that interacts with its environment through structured function-calling APIs~\cite{yao2022react,schick2023toolformer}. The attacker controls untrusted content returned by these tools, such as poisoned webpages, malicious emails, retrieved files, or compromised API responses~\cite{greshake2023not,zhan2403injecagent}. The attack goal is \emph{indirect prompt injection}: to manipulate replanning so that untrusted content determines an authority-bearing tool argument, enabling data exfiltration, target hijacking, file overwrite, or code execution.

\pact{} provides an enforcement-layer guarantee. Its trusted computing base consists of the runtime monitor, tool contracts, the provenance-preserving agent framework, and trusted discharge procedures used at \lthree{}. We assume conservative provenance propagation and contracts that encode the intended authority semantics of each tool argument. Errors in automatic contract synthesis or provenance inference are not part of the oracle guarantee; they are measured separately in deployment. PACT also does not address tool-selection attacks, where the model invokes a dangerous tool using trusted constants, or content-channel attacks, where harmful text appears inside an otherwise permitted \content{} argument.

\subsection{Argument-Level Contracts}
\label{sec:contracts}

For a tool $t$ with $k$ arguments, \pact{} defines a contract
\[
C_t=(\ell,\{a_i\}_{i=1}^{k},o),
\]
where $\ell\in\{\lzero,\lone,\ltwo,\lthree\}$ is the contract precision level and $o$ is an \texttt{OutputSpec} assigning provenance to tool outputs. Each argument entry is
\[
a_i=(\mathrm{name}_i,\mathrm{role}_i,\tau_i^{\min},F_i,\mathcal{R}_i,\mathcal{D}_i),
\]
where $\mathrm{role}_i$ is the semantic role, $\tau_i^{\min}$ is the minimum trust required, $F_i$ is a forbidden-origin set, $\mathcal{R}_i$ is the set of required obligations, and $\mathcal{D}_i$ is the set of trusted discharge procedures available at \lthree{}.

\paragraph{Roles and trust.}
PACT uses six roles: \target{} for authority-bearing destinations such as recipients, URLs, or endpoints; \command{} for executable commands or queries; \credential{} for secrets; \content{} for payload text; \texttt{selector} for object selection; and \texttt{control} for behavior-modifying flags. These roles are not a linguistic taxonomy. They are the policy interface that separates arguments that \emph{bind authority} from arguments that primarily \emph{carry content}. Trust values form the ordered lattice
\[
\trusted \;>\; \usertr \;>\; \toolout \;>\; \external .
\]

Contract precision controls how much argument structure is exposed:
\begin{itemize}[nosep,leftmargin=*]
    \item \lzero{} (\emph{Opaque}): treat the tool call as one authority boundary.
    \item \lone{} (\emph{Capability}): enforce one threshold for all arguments of a capability.
    \item \ltwo{} (\emph{Argument-role}): check each argument against its role-specific requirement.
    \item \lthree{} (\emph{Certified routing}): apply \ltwo{} checks, but allow explicitly dischargeable violations to pass through trusted procedures in $\mathcal{D}_i$.
\end{itemize}
Coarser contracts remain safe when argument semantics are unavailable; finer contracts recover benign mixed-trust workflows once those semantics are specified.

\paragraph{Trusted discharge.}
Certified routing is not an unrestricted trust upgrade. A trusted procedure returns a scoped certificate
\[
d=(r,s,\rho,\tau^{\max}),
\]
where $r$ is the failed predicate or obligation being discharged, $s$ is the trusted procedure, $\rho$ is the argument-role scope, and $\tau^{\max}$ bounds the effective trust of the certified value. Applying $d$ to $v$ creates a derived value
\[
\pi(v^{d})=
\prov{\origins(v)\cup\{s\},\; \min(\tau^{\max},\tau_{\mathrm{cert}}(s)),\; \oblig(v)\setminus\{r\}} .
\]
The original origins are retained, only the scoped obligation is discharged, and the certificate cannot be reused for roles outside $\rho$.

\subsection{Provenance and Runtime Checking}
\label{sec:monitor}

Every runtime value $v$ carries a provenance tag
\[
\pi(v)=\prov{\origins(v),\trust(v),\oblig(v)},
\]
where $\origins(v)$ is the set of contributing sources, $\trust(v)$ is the current trust level, and $\oblig(v)$ is the set of unresolved obligations. User instructions receive $\usertr$, trusted constants receive $\trusted$, and tool outputs receive the trust specified by their \texttt{OutputSpec}. When values are combined, provenance is merged conservatively:
\begin{equation}
\mathrm{merge}\!\left(
\prov{O_1,\tau_1,B_1},
\prov{O_2,\tau_2,B_2}
\right)
=
\prov{O_1\cup O_2,\min(\tau_1,\tau_2),B_1\cup B_2}.
\label{eq:merge}
\end{equation}
Thus ordinary dataflow never increases trust and never erases origins or obligations. Higher effective trust can arise only through an explicit \lthree{} discharge certificate, which creates a bounded, role-scoped derived value.

Before executing $t(v_1,\ldots,v_k)$, the monitor checks each argument $v_i$ against $a_i$:
\[
\trust(v_i)\ge \tau_i^{\min},
\qquad
\origins(v_i)\cap F_i=\emptyset,
\qquad
\oblig(v_i)\cup \mathcal{R}_i \subseteq \mathsf{Discharged}.
\]
At \lthree{}, a failed predicate may be discharged only if some $d\in\mathcal{D}_i$ covers the failed predicate and the argument role; otherwise the call is blocked.

\vspace*{-0.2cm}
\begin{algorithm}[H]
\renewcommand{\arraystretch}{0.9}
\caption{\pact{} checks each argument against its role-specific contract before execution.}
\label{alg:pact-check}
\begin{algorithmic}[1]
\Require Tool call $t(v_1,\ldots,v_k)$, contract $C_t=(\ell,\{a_i\}_{i=1}^{k},o)$, provenance state $\pi$
\For{$i=1$ \textbf{to} $k$}
    \State read $a_i=(\mathrm{name}_i,\mathrm{role}_i,\tau_i^{\min},F_i,\mathcal{R}_i,\mathcal{D}_i)$
    \State $\mathcal{V}_i \gets$ failed trust, origin, or obligation predicates
    \If{$\mathcal{V}_i\neq \emptyset$}
        \If{$\ell=\lthree$ and some valid $d\in\mathcal{D}_i$ covers $\mathcal{V}_i$ and $\mathrm{role}_i$}
            \State replace $v_i$ by the certified derived value $v_i^d$
        \Else
            \State \Return \textsc{Block}
        \EndIf
    \EndIf
\EndFor
\State execute $t(v_1,\ldots,v_k)$ and tag outputs according to $o$
\State \Return \textsc{Allow}
\end{algorithmic}
\vspace*{-0.1cm}
\end{algorithm}
\vspace*{-0.4cm}

The online check is linear in the number of arguments plus the size of the origin, obligation, and discharge sets. It performs only lattice comparisons, set intersections, and certificate-scope checks at execution time.

\subsection{Automatic Contract and Provenance Inference}
\label{sec:auto_inference}

The formal monitor assumes contracts and provenance tags. For deployment, we instantiate them with an automatic, fail-closed pipeline; its fidelity explains the oracle-to-deployment gap. Given a tool schema, we synthesize a draft contract from the tool name, argument names, type annotations, and natural-language descriptions. Deterministic rules assign high-confidence roles: recipients, URLs, endpoints, paths, attendees, accounts, and destinations map to \target{}; shell commands, executable queries, and mutation requests map to \command{}; API keys, tokens, and passwords map to \credential{}; message bodies, summaries, reports, and quoted text map to \content{}. Ambiguous fields such as \texttt{query}, \texttt{selector}, or \texttt{filter} are assigned conservatively unless the schema makes their authority semantics clear; missing specifications fail closed at \ltwo{}.

Each \texttt{OutputSpec} assigns default trust and origin policies to returned values. External retrieval tools produce $\external$ unless a trusted tool-specific attestation states otherwise; we use fail-low defaults for tools that read webpages, files, messages, or external APIs. At runtime, provenance is inferred by exact structural matching, role-aware heuristics for high-confidence transformations, and an LLM classifier for remaining ambiguous arguments. Low-confidence or authority-bearing cases fall back to the more conservative provenance label. In deployment, this pipeline synthesizes 74 contracts and achieves 87.1\% role accuracy and 77.4\% provenance accuracy on 20 real MCP tools. These numbers are not assumptions of the formal model; they quantify the upstream fidelity that limits fully automatic deployment.

\subsection{Formal Properties}
\label{sec:formal_properties}

The formal role of \pact{} is to isolate the enforcement layer. The statements below assume conservative provenance propagation, correctly specified contracts, and scoped non-expansive discharge certificates; they do not assume perfect automatic inference.

\begin{definition}[Flat tool-level monitor]
\label{def:flat}
A \emph{flat tool-level monitor} assigns each tool $t$ one label $\lambda_t\in\{\textsc{High},\textsc{Low}\}$ and makes one allow/block decision for the entire invocation. It blocks $t(v_1,\ldots,v_k)$ iff $\lambda_t=\textsc{High}$ and some argument $v_i$ has $\trust(v_i)<\tau_{\textsc{High}}$.
\end{definition}

\begin{definition}[Safe contract refinement]
\label{def:safe-refinement}
A contract $C'$ is a \emph{safe refinement} of $C$ if it exposes finer argument roles without allowing any origin to bind an authority-bearing argument that $C$ would forbid. At \lthree{}, a refinement may admit an execution only through a scoped discharge certificate whose derived value preserves original origins and is valid only for roles within its scope.
\end{definition}

\begin{theorem}[Separation from flat enforcement]
\label{thm:separation}
There exist mixed-trust tool environments in which role-aware argument contracts achieve full benign utility and perfect argument-provenance integrity, whereas any flat tool-level monitor incurs either a false positive or a false negative.
\end{theorem}

\begin{theorem}[Precision monotonicity under safe refinement]
\label{thm:monotonicity}
Let $C_0,\ldots,C_3$ be contracts ordered by safe refinement from \lzero{} to \lthree{}, with fixed provenance semantics and non-expansive discharge certificates. Then refinement preserves the authority-binding security property while monotonically reducing unnecessary blocking:
\[
\mathrm{blocked}(M_0)\supseteq \cdots \supseteq \mathrm{blocked}(M_3),
\qquad
\mathrm{allowed}(M_0)\subseteq \cdots \subseteq \mathrm{allowed}(M_3).
\]
\end{theorem}

\begin{theorem}[Prefix-based soundness]
\label{thm:prefix}
For any execution prefix $\tau_k=(c_1,\ldots,c_k)$, if the \pact{} monitor allows $c_k$ given the provenance state induced by $\tau_{k-1}$, then the provenance invariants for $c_k$ hold independently of future actions $c_{k+1},\ldots,c_n$.
\end{theorem}

\Cref{thm:separation} captures the granularity mismatch: whole-call monitors must over-block benign mixed-trust workflows or miss authority hijacking. \Cref{thm:monotonicity} states that exposing argument roles can recover utility without relaxing the authority-binding constraint. \Cref{thm:prefix} justifies dynamic replanning: each decision depends only on provenance accumulated so far. Full proofs appear in \cref{app:proofs}. The oracle experiments test this enforcement rule; deployment additionally tests contract synthesis, \texttt{OutputSpec} assignment, and provenance inference.

\section{Experiments}
\label{sec:experiments}

Our experiments test the authority-binding claim at four levels. First, we isolate the mechanism under oracle provenance: does argument-level enforcement remove the false-positive/false-negative tradeoff induced by whole-call policies? Second, we evaluate full deployment in AgentDojo, where contracts and provenance are inferred automatically and tools execute through the real runtime. Third, we ablate the mechanisms that distinguish PACT from ordinary per-argument filtering or taint tracking. Finally, we characterize practicality and scope: runtime overhead, external stress tests, and the failure modes that remain outside structural provenance enforcement.

\subsection{Mechanism: Is Granularity the Right Abstraction?}
\label{sec:exp-mechanism}
\begin{wraptable}[14]{r}{6cm}
\vspace{-0.5cm}
\small
\centering
\renewcommand{\arraystretch}{0.95}\tabcolsep 0.1cm
\caption{\textbf{Argument-level enforcement removes the mechanism-level tradeoff.}
On mixed-trust diagnostic scenarios, PACT is the only defense that achieves both full benign utility and full security. Invocation-level defenses either over-block benign retrieval-then-act workflows or fail to block authority hijacking.}
\label{tab:mechanism}
\begin{tabular}{lccccc}
\toprule
Defense & Utility & Security & FP$\downarrow$ & FN$\downarrow$ & Acc. \\
\midrule
Vanilla & 100.0 & 0.0 & 0 & 8 & 52.9 \\
FIDES & 66.7 & 37.5 & 3 & 5 & 52.9 \\
CaMeL & 33.3 & 100.0 & 6 & 0 & 64.7 \\
PACT L2 & \textbf{100.0} & \textbf{100.0} & \textbf{0} & \textbf{0} & \textbf{100.0} \\
\bottomrule
\end{tabular}
\end{wraptable}

\paragraph{Setup.}
We first evaluate the enforcement abstraction under matched tool interfaces and oracle provenance. The diagnostic suite contains 17 mixed-trust scenarios across email, messaging, file I/O, code execution, API calls, and calendar workflows. Each scenario is constructed to contain both a benign mixed-trust action and an attack in which untrusted content attempts to bind an authority-bearing argument. We compare PACT against vanilla execution and mechanism-level reimplementations of FIDES~\cite{costa2505securing}  and CaMeL~\cite{debenedetti2025defeating}. These reimplementations are used to compare enforcement abstractions, not to reproduce each system's full agent stack.

\paragraph{Result.}
Table~\ref{tab:mechanism} shows the central mechanism result. PACT L2 achieves 100\% utility and 100\% security, with no false positives and no false negatives. The baselines fail in complementary ways. Vanilla execution preserves benign utility but allows all attacks. CaMeL blocks attacks but also blocks benign workflows that require external content to fill non-authority arguments. FIDES lies between these extremes, but still both over-blocks benign mixed-trust cases and misses attacks where tool-derived values reach authority-bearing arguments.

\paragraph{Interpretation.}
This result is the empirical counterpart of the separation theorem. The tradeoff is not between ``using external content'' and ``being secure''; it is between enforcing trust at the wrong granularity and enforcing it at the argument where authority is bound. To check that the diagnostic suite is not a small hand-picked set, we further evaluate 50 additional controlled scenarios, 25 white-box attacks, 30 held-out attacks, and 75 replanning stress tests at 5/10/20 steps. PACT remains correct throughout these stress tests; we report the full breakdown in Appendix~\ref{app:mechanism_stress}.

\begin{wrapfigure}[20]{r}{6cm}
\vspace{-0.5cm}
\centering
\includegraphics[width=1\linewidth]{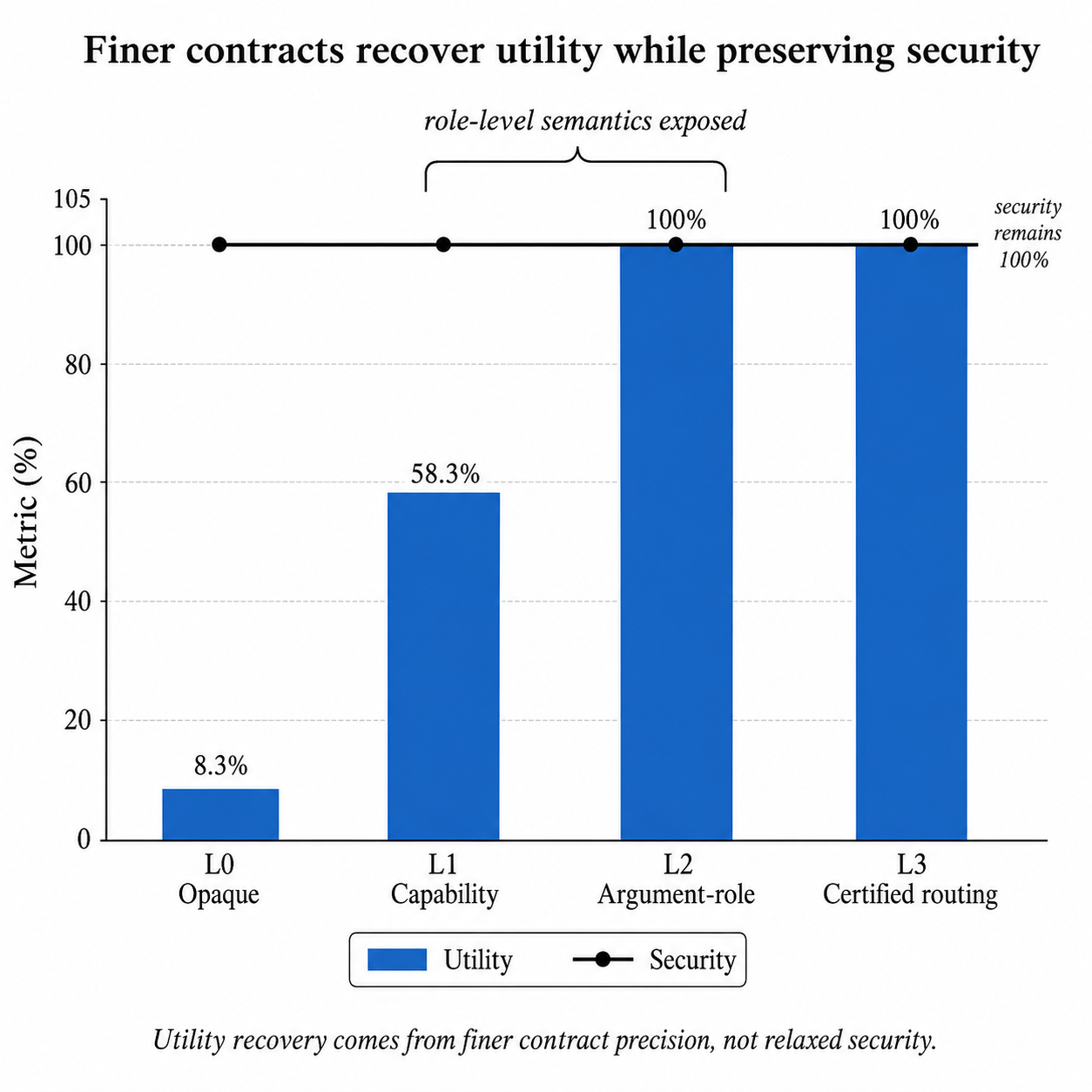}
\caption{\textbf{Contract precision recovers benign utility without relaxing security.}
Under oracle provenance, security remains at 100\% across L0--L3, while utility rises as contracts move from opaque blocking to role-aware argument checks.}
\label{fig:precision}
\end{wrapfigure}

\paragraph{Contract precision.}
We next test the monotonicity prediction of the contract hierarchy. Under fixed oracle provenance, security remains at 100\% across L0--L3, while utility increases from 8.3\% at L0 to 58.3\% at L1 and 100\% at L2/L3. Figure~\ref{fig:precision} visualizes this trend. Finer contracts recover benign behavior by exposing argument roles; they do not relax the prohibition on untrusted data binding authority-bearing arguments.

\subsection{Deployment: Does PACT Improve High-Security Utility?}
\label{sec:exp-agentdojo}

\paragraph{Setup.}
We next evaluate PACT in the full AgentDojo v1 runtime~\cite{debenedetti2024agentdojo}. All defenses are integrated as \texttt{BasePipelineElement}s and tools execute through the actual \texttt{FunctionsRuntime}. We use all 97 benign user tasks and all 27 injection tasks across banking, workspace, Slack, and travel. We evaluate five models from three families and use matched denominators for every defense. PACT is deployed through a fully automatic pipeline that synthesizes 74 tool contracts from schemas and infers provenance using exact structural matching, role-aware heuristics, and an LLM classifier for ambiguous arguments.

\paragraph{Result.}
Table~\ref{tab:agentdojo} shows that PACT is not best understood as maximizing raw utility regardless of risk. NoDefense and FIDES sometimes complete more benign tasks, but they operate at lower security. The relevant comparison is the high-security regime. On Qwen-turbo and Qwen-plus, PACT reaches 96.3\% security while preserving substantially more utility than CaMeL. On Qwen-max, Qwen2.5-72B, and GPT-4o-mini, PACT reaches 100\% security while achieving 38.1--46.4\% utility, improving over CaMeL by 8--16 percentage points at the same security level.

\begin{wraptable}[16]{r}{8cm}
\vspace{-0.5cm}
\small
\renewcommand{\arraystretch}{0.94}\tabcolsep 0.1cm
\centering
\caption{\textbf{Full AgentDojo deployment: PACT improves utility in the high-security regime.}
Bold marks the highest utility among defenses reaching at least 96\% security on that model. On the three strongest models, PACT reaches 100\% security while recovering 8--16 percentage points more utility than CaMeL.}
\label{tab:agentdojo}
\begin{tabular}{llcccc}
\toprule
Model & Metric & NoDefense & FIDES & CaMeL & PACT \\
\midrule
Qwen-turbo & Utility & 20.6 & 19.6 & 11.3 & \textbf{18.6} \\
           & Security & 92.6 & 88.9 & 100.0 & 100.0 \\
Qwen-plus  & Utility & 60.8 & 59.8 & 29.9 & \textbf{45.4} \\
           & Security & 88.9 & 85.2 & 100.0 & 100.0 \\
Qwen-max   & Utility & 46.4 & 50.5 & 29.9 & \textbf{38.1} \\
           & Security & 96.3 & 96.3 & 100.0 & 100.0 \\
Qwen2.5-72B & Utility & 46.4 & 48.5 & 28.9 & \textbf{42.3} \\
            & Security & 96.3 & 96.3 & 100.0 & 100.0 \\
GPT-4o-mini & Utility & 59.8 & 55.7 & 30.9 & \textbf{46.4} \\
            & Security & 88.9 & 88.9 & 100.0 & 100.0 \\
\bottomrule
\end{tabular}
\end{wraptable}

\paragraph{Interpretation.}
These results show that the mechanism survives deployment when the base agent is strong enough to complete benign workflows. Invocation-level quarantining can reach perfect security, but it discards many legitimate retrieval-then-act actions. Argument-level enforcement recovers a substantial fraction of that utility by allowing external data in content roles while blocking it from authority-bearing roles. The remaining gap between oracle and deployed results is expected: automatic role assignment and provenance inference introduce errors that do not exist in the oracle mechanism setting.

\subsection{Mechanism Ablations: What Matters?}
\label{sec:exp-ablation}

\paragraph{Roles are necessary, not cosmetic.}
A natural alternative is to check each argument individually using a uniform trust threshold, without semantic roles. Table~\ref{tab:ablation} shows that this is insufficient. A uniform PerArg-USER policy reaches 100\% security on AgentDojo but drops utility to 14.4\%. A uniform PerArg-EXTERNAL policy recovers 18.6\% utility but falls to 92.6\% security. PACT matches the higher utility while improving security to 96.3\%. The gain therefore does not come from checking arguments individually; it comes from assigning different trust requirements to different argument roles.

\begin{wraptable}[17]{r}{5cm}

\small
\renewcommand{\arraystretch}{0.95}\tabcolsep 0.1cm
\centering
\caption{\textbf{Roles and cross-step provenance address distinct failure modes.}
Uniform per-argument thresholds cannot recover both utility and security. Removing roles collapses benign utility; removing cross-step origin tracking creates security failures.}
\label{tab:ablation}
\begin{tabular}{lcc}
\toprule
Configuration & Utility & Security \\
\midrule
\multicolumn{3}{l}{\emph{Role-agnostic AgentDojo ablation}} \\
PerArg-USER & 14.4 & 100.0 \\
PerArg-EXTERNAL & 18.6 & 92.6 \\
PACT with roles & \textbf{18.6} & 96.3 \\
\midrule
\multicolumn{3}{l}{\emph{Controlled component ablation}} \\
PACT-full & \textbf{100.0} & \textbf{100.0} \\
No-roles (L1) & 40.0 & 100.0 \\
No-crossstep & 100.0 & 80.0 \\
No-policies & 100.0 & 100.0 \\
\bottomrule
\end{tabular}
\end{wraptable}

\paragraph{Cross-step provenance prevents laundering.}
The controlled component ablation isolates the other key mechanism. Removing semantic roles collapses utility by 60 points with no security benefit. Removing cross-step origin tracking leaves benign utility unchanged but reduces security from 100\% to 80\%. These failures correspond to different parts of the method: roles resolve benign mixed-trust workflows, while accumulated origins prevent untrusted content from being laundered through intermediate tool calls.

\subsection{Practicality, Stress Tests, and Boundaries}
\label{sec:exp-practical}

\paragraph{Runtime overhead is not the bottleneck.}
PACT's online monitor performs set membership checks and lattice comparisons before tool execution. Across 100K checks, all PACT levels have P99 latency below 272 microseconds and throughput between 13.6K and 18.0K checks/s. This cost is small relative to LLM API latency; the observed deployment bottleneck is not runtime checking, but the fidelity of contracts and provenance tags.

\paragraph{Residual errors are attributable.}
In the automatic deployment pipeline, contract-role inference reaches 87.1\% accuracy and provenance inference reaches 77.4\% accuracy on 20 real MCP tools. In deployed AgentDojo runs, role errors account for false blocks, while provenance errors primarily affect web-to-email chains; the most important remaining failure mode is output-trust over-labeling. This is a limitation, but also a diagnostic result: once the enforcement rule is fixed, residual failures can be attributed to upstream role and provenance fidelity rather than to ambiguity in the runtime policy.

\paragraph{External stress tests.}
We place broader stress tests in the appendix because they characterize generality and boundaries rather than the core mechanism. On a five-domain cross-domain benchmark covering healthcare, e-commerce, social media, DevOps, and smart home tools, PACT achieves 100\% utility and 100\% security on 25 manually annotated scenarios. In a 7-day deployment study with 280 calls across 15 tools, PACT blocks all 13 observed attacks with zero false positives on 267 benign calls. On InjecAgent~\cite{zhan2403injecagent}, PACT and CaMeL both block all 1,598 attack tool calls, while FIDES misses a substantial fraction of data-stealing cases. On Agent-SafetyBench~\cite{zhang2024agent}, PACT blocks 74.7\% of unsafe cases with 39.3\% benign utility; the misses are concentrated in browser UI interactions and pure information-retrieval cases, which fall partly outside structural provenance enforcement. These results support the scope of PACT without changing the main claim: PACT is a defense for authority binding, not a general content-safety filter.

\section{Conclusion}

PACT reframes indirect prompt injection~\cite{greshake2023not,liu2023prompt} as an authority-binding problem. The central question is not whether an agent has encountered untrusted content, but whether that content is permitted to determine an authority-bearing argument. By moving enforcement from whole tool invocations to role-specific arguments, PACT separates benign use of external information from privileged control over targets, commands, and credentials.

This granularity shift is supported by both mechanism and deployment evidence. Under oracle provenance, PACT achieves 100\% utility and 100\% security on mixed-trust diagnostic suites; in full AgentDojo~\cite{debenedetti2024agentdojo} deployment, it reaches 100\% security on the strongest models while recovering higher utility than invocation-level quarantining~\cite{debenedetti2025defeating,costa2505securing} in the high-security regime. The oracle-to-deployment gap is also informative: once the enforcement rule is fixed, the remaining bottleneck shifts to provenance inference and contract synthesis. PACT therefore suggests that part of the agent security--utility tradeoff is not inherent to using external information, but to enforcing trust at the wrong granularity.

\begin{ack}
\end{ack}
\newpage

\bibliographystyle{plainnat}
\bibliography{references}

@article{debenedetti2025defeating,
  title={Defeating prompt injections by design},
  author={Debenedetti, Edoardo and Shumailov, Ilia and Fan, Tianqi and Hayes, Jamie and Carlini, Nicholas and Fabian, Daniel and Kern, Christoph and Shi, Chongyang and Terzis, Andreas and Tram{\`e}r, Florian},
  journal={arXiv preprint arXiv:2503.18813},
  year={2025}
}

@article{ji2026taming,
  title={Taming Various Privilege Escalation in LLM-Based Agent Systems: A Mandatory Access Control Framework},
  author={Ji, Zimo and Wu, Daoyuan and Jiang, Wenyuan and Ma, Pingchuan and Li, Zongjie and Gao, Yudong and Wang, Shuai and Li, Yingjiu},
  journal={arXiv preprint arXiv:2601.11893},
  year={2026}
}

@article{chen2026learning,
  title={Learning to Inject: Automated Prompt Injection via Reinforcement Learning},
  author={Chen, Xin and Zhang, Jie and Tramer, Florian},
  journal={arXiv preprint arXiv:2602.05746},
  year={2026}
}

@article{debenedetti2024agentdojo,
  title={Agentdojo: A dynamic environment to evaluate prompt injection attacks and defenses for llm agents},
  author={Debenedetti, Edoardo and Zhang, Jie and Balunovic, Mislav and Beurer-Kellner, Luca and Fischer, Marc and Tram{\`e}r, Florian},
  journal={Advances in Neural Information Processing Systems},
  volume={37},
  pages={82895--82920},
  year={2024}
}

@article{zhan2403injecagent,
  title={InjecAgent: Benchmarking indirect prompt injections in tool-integrated large language model agents, 2024},
  author={Zhan, Qiusi and Liang, Zhixiang and Ying, Zifan and Kang, Daniel},
  journal={URL: https://arxiv.org/abs/2403.02691}
}

@article{costa2505securing,
  title={Securing ai agents with information-flow control, 2025},
  author={Costa, Manuel and K{\"o}pf, Boris and Kolluri, Aashish and Paverd, Andrew and Russinovich, Mark and Salem, Ahmed and Tople, Shruti and Wutschitz, Lukas and Zanella-B{\'e}guelin, Santiago},
  journal={URL https://arxiv.org/abs/2505.23643}
}

@article{schick2023toolformer,
  title={Toolformer: Language models can teach themselves to use tools},
  author={Schick, Timo and Dwivedi-Yu, Jane and Dess{\`\i}, Roberto and Raileanu, Roberta and Lomeli, Maria and Hambro, Eric and Zettlemoyer, Luke and Cancedda, Nicola and Scialom, Thomas},
  journal={Advances in neural information processing systems},
  volume={36},
  pages={68539--68551},
  year={2023}
}

@article{patil2024gorilla,
  title={Gorilla: Large language model connected with massive apis},
  author={Patil, Shishir G and Zhang, Tianjun and Wang, Xin and Gonzalez, Joseph E},
  journal={Advances in Neural Information Processing Systems},
  volume={37},
  pages={126544--126565},
  year={2024}
}

@article{yao2022react,
  title={React: Synergizing reasoning and acting in language models},
  author={Yao, Shunyu and Zhao, Jeffrey and Yu, Dian and Du, Nan and Shafran, Izhak and Narasimhan, Karthik and Cao, Yuan},
  journal={arXiv preprint arXiv:2210.03629},
  year={2022}
}

@article{liu2023prompt,
  title={Prompt injection attack against llm-integrated applications},
  author={Liu, Yi and Deng, Gelei and Li, Yuekang and Wang, Kailong and Wang, Zihao and Wang, Xiaofeng and Zhang, Tianwei and Liu, Yepang and Wang, Haoyu and Zheng, Yan and others},
  journal={arXiv preprint arXiv:2306.05499},
  year={2023}
}

@article{xiang2024guardagent,
  title={Guardagent: Safeguard llm agents by a guard agent via knowledge-enabled reasoning},
  author={Xiang, Zhen and Zheng, Linzhi and Li, Yanjie and Hong, Junyuan and Li, Qinbin and Xie, Han and Zhang, Jiawei and Xiong, Zidi and Xie, Chulin and Yang, Carl and others},
  journal={arXiv preprint arXiv:2406.09187},
  year={2024}
}

@inproceedings{liu2024formalizing,
  title={Formalizing and benchmarking prompt injection attacks and defenses},
  author={Liu, Yupei and Jia, Yuqi and Geng, Runpeng and Jia, Jinyuan and Gong, Neil Zhenqiang},
  booktitle={33rd USENIX Security Symposium (USENIX Security 24)},
  pages={1831--1847},
  year={2024}
}

@article{zou2023universal,
  title={Universal and transferable adversarial attacks on aligned language models},
  author={Zou, Andy and Wang, Zifan and Carlini, Nicholas and Nasr, Milad and Kolter, J Zico and Fredrikson, Matt},
  journal={arXiv preprint arXiv:2307.15043},
  year={2023}
}

@inproceedings{tiwari2024information,
  title={Information flow control in machine learning through modular model architecture},
  author={Tiwari, Trishita and Gururangan, Suchin and Guo, Chuan and Hua, Weizhe and Kariyappa, Sanjay and Gupta, Udit and Xiong, Wenjie and Maeng, Kiwan and Lee, Hsien-Hsin S and Suh, G Edward},
  booktitle={33rd USENIX Security Symposium (USENIX Security 24)},
  pages={6921--6938},
  year={2024}
}

@article{jin2026capseal,
  title={CapSeal: Capability-Sealed Secret Mediation for Secure Agent Execution},
  author={Jin, Shutong and Guo, Ruiyi and Cheung, Ray CC},
  journal={arXiv preprint arXiv:2604.16762},
  year={2026}
}

@article{guo2026ih,
  title={IH-Challenge: A Training Dataset to Improve Instruction Hierarchy on Frontier LLMs},
  author={Guo, Chuan and Uribe, Juan Felipe Ceron and Zhu, Sicheng and Choquette-Choo, Christopher A and Lin, Steph and Kandpal, Nikhil and Nasr, Milad and Toyer, Sam and Wang, Miles and Yu, Yaodong and others},
  journal={arXiv preprint arXiv:2603.10521},
  year={2026}
}

@article{he2026taintp2x,
  title={TaintP2X: Detecting Taint-Style Prompt-to-Anything Injection Vulnerabilities in LLM-Integrated Applications},
  author={He, Junjie and Wang, Shenao and Zhao, Yanjie and Hou, Xinyi and Liu, Zhao and Zou, Quanchen and Wang, Haoyu},
  year={2026}
}

@inproceedings{greshake2023not,
  title={Not what you've signed up for: Compromising real-world llm-integrated applications with indirect prompt injection},
  author={Greshake, Kai and Abdelnabi, Sahar and Mishra, Shailesh and Endres, Christoph and Holz, Thorsten and Fritz, Mario},
  booktitle={Proceedings of the 16th ACM workshop on artificial intelligence and security},
  pages={79--90},
  year={2023}
}

@article{wang2026landscape,
  title={The Landscape of Prompt Injection Threats in LLM Agents: From Taxonomy to Analysis},
  author={Wang, Peiran and Li, Xinfeng and Xiang, Chong and Zhang, Jinghuai and Li, Ying and Zhang, Lixia and Wang, Xiaofeng and Tian, Yuan},
  journal={arXiv preprint arXiv:2602.10453},
  year={2026}
}

@article{shi2504progent,
  title={Progent: Programmable privilege control for llm agents, 2025},
  author={Shi, Tianneng and He, Jingxuan and Wang, Zhun and Li, Hongwei and Wu, Linyu and Guo, Wenbo and Song, Dawn},
  journal={URL https://arxiv.org/abs/2504.11703}
}

@article{sequeira2026agent,
  title={Agent-Sentry: Bounding LLM Agents via Execution Provenance},
  author={Sequeira, Rohan and Damianakis, Stavros and Iqbal, Umar and Psounis, Konstantinos},
  journal={arXiv preprint arXiv:2603.22868},
  year={2026}
}

@article{wang2025agentarmor,
  title={Agentarmor: Enforcing program analysis on agent runtime trace to defend against prompt injection},
  author={Wang, Peiran and Liu, Yang and Lu, Yunfei and Cai, Yifeng and Chen, Hongbo and Yang, Qingyou and Zhang, Jie and Hong, Jue and Wu, Ye},
  journal={arXiv preprint arXiv:2508.01249},
  year={2025}
}

@article{zhu2025melon,
  title={Melon: Provable defense against indirect prompt injection attacks in ai agents},
  author={Zhu, Kaijie and Yang, Xianjun and Wang, Jindong and Guo, Wenbo and Wang, William Yang},
  journal={arXiv preprint arXiv:2502.05174},
  year={2025}
}

@article{li2025drift,
  title={Drift: Dynamic rule-based defense with injection isolation for securing llm agents},
  author={Li, Hao and Liu, Xiaogeng and Chiu, Hung-Chun and Li, Dianqi and Zhang, Ning and Xiao, Chaowei},
  journal={arXiv preprint arXiv:2506.12104},
  year={2025}
}

@article{liu2026safeagent,
  title={SafeAgent: A Runtime Protection Architecture for Agentic Systems},
  author={Liu, Hailin and Ilyushin, Eugene and Ni, Jie and Zhu, Min},
  journal={arXiv preprint arXiv:2604.17562},
  year={2026}
}

@article{wallace2024instruction,
  title={The instruction hierarchy: Training llms to prioritize privileged instructions},
  author={Wallace, Eric and Xiao, Kai and Leike, Reimar and Weng, Lilian and Heidecke, Johannes and Beutel, Alex},
  journal={arXiv preprint arXiv:2404.13208},
  year={2024}
}

@article{wu2024instructional,
  title={Instructional segment embedding: Improving llm safety with instruction hierarchy},
  author={Wu, Tong and Zhang, Shujian and Song, Kaiqiang and Xu, Silei and Zhao, Sanqiang and Agrawal, Ravi and Indurthi, Sathish Reddy and Xiang, Chong and Mittal, Prateek and Zhou, Wenxuan},
  journal={arXiv preprint arXiv:2410.09102},
  year={2024}
}

@article{evtimov2025wasp,
  title={Wasp: Benchmarking web agent security against prompt injection attacks},
  author={Evtimov, Ivan and Zharmagambetov, Arman and Grattafiori, Aaron and Guo, Chuan and Chaudhuri, Kamalika},
  journal={arXiv preprint arXiv:2504.18575},
  year={2025}
}

@article{sun2025seagent,
  title={Seagent: Self-evolving computer use agent with autonomous learning from experience},
  author={Sun, Zeyi and Liu, Ziyu and Zang, Yuhang and Cao, Yuhang and Dong, Xiaoyi and Wu, Tong and Lin, Dahua and Wang, Jiaqi},
  journal={arXiv preprint arXiv:2508.04700},
  year={2025}
}

@article{zhu2025scaling,
  title={Scaling test-time compute for llm agents},
  author={Zhu, King and Li, Hanhao and Wu, Siwei and Xing, Tianshun and Ma, Dehua and Tang, Xiangru and Liu, Minghao and Yang, Jian and Liu, Jiaheng and Jiang, Yuchen Eleanor and others},
  journal={arXiv preprint arXiv:2506.12928},
  year={2025}
}

@article{wang2026adaptools,
  title={AdapTools: Adaptive tool-based indirect prompt injection attacks on agentic LLMs},
  author={Wang, Che and Zhang, Jiaming and Zhang, Ziqi and Wang, Zijie and Wang, Yinghui and Gao, Jianbo and Wei, Tao and Chen, Zhong and Lim, Wei Yang Bryan},
  journal={arXiv preprint arXiv:2602.20720},
  year={2026}
}

@article{lan2026cognitive,
  title={The Cognitive Firewall: Securing Browser Based AI Agents Against Indirect Prompt Injection Via Hybrid Edge Cloud Defense},
  author={Lan, Qianlong and Kaul, Anuj},
  journal={arXiv preprint arXiv:2603.23791},
  year={2026}
}

@article{shi2025prompt,
  title={Prompt injection attack to tool selection in llm agents},
  author={Shi, Jiawen and Yuan, Zenghui and Tie, Guiyao and Zhou, Pan and Gong, Neil Zhenqiang and Sun, Lichao},
  journal={arXiv preprint arXiv:2504.19793},
  year={2025}
}

@article{lee2024prompt,
  title={Prompt infection: Llm-to-llm prompt injection within multi-agent systems},
  author={Lee, Donghyun and Tiwari, Mo},
  journal={arXiv preprint arXiv:2410.07283},
  year={2024}
}

@article{zhang2024agent,
  title={Agent-safetybench: Evaluating the safety of llm agents},
  author={Zhang, Zhexin and Cui, Shiyao and Lu, Yida and Zhou, Jingzhuo and Yang, Junxiao and Wang, Hongning and Huang, Minlie},
  journal={arXiv preprint arXiv:2412.14470},
  year={2024}
}

\newpage

\appendix

\section{Limitations.}
PACT gives an enforcement-layer guarantee: assuming conservative provenance propagation and correctly specified contracts, the runtime can decide whether a source is permitted to bind a given argument role. Its deployed behavior is therefore limited by the fidelity of the contract and provenance layer. In our automatic pipeline, role assignment reaches 87.1\% accuracy and provenance inference reaches 77.4\% accuracy; the dominant residual errors arise from these upstream components rather than from the runtime rule itself. This is the main practical limitation of the current system, but it is also a useful diagnostic separation: after the enforcement principle is fixed, failures become attributable to role and provenance fidelity.

A second limitation is semantic ambiguity at the argument interface. Arguments such as \texttt{query} may behave as content in one tool, a selector in another, and a command-like control field in a third. PACT resolves such cases conservatively, which preserves the authority-binding constraint but can over-block benign workflows and leave utility on the table.

Finally, PACT is a structural defense, not a general solution to agent safety. It does not address content-channel attacks, where harmful content appears inside an otherwise permitted \content{} field, nor tool-selection attacks, where the model invokes a dangerous tool using trusted constants. Broad safety benchmarks~\cite{zhang2024agent,evtimov2025wasp} that include browser UI manipulation or pure information-retrieval harms therefore measure risks partly outside PACT's provenance-enforcement scope. Our empirical evaluation covers three model families; broader validation on additional proprietary and open-weight agents remains important future work.

\section{Proof Details}
\label{app:proofs}

This section gives the full proofs for the three formal claims used in the main text. These are enforcement-layer statements: they assume the monitor family in \cref{def:flat}, the contract hierarchy in \cref{sec:contracts}, and the provenance semantics in \cref{sec:monitor}. In particular, the monotonicity result compares monitors over the same tool interface and \texttt{OutputSpec}; the monitors differ only in the precision of argument checks and in whether trusted routing is available.

\paragraph{Why this matters.}
The main text uses theorems to isolate the granularity problem. The proofs below make explicit that the separation is not a statement about all agent security, but about the insufficiency of whole-call enforcement for mixed-trust tool invocations.

\subsection{Separation Intuition}
\label{app:separation_intuition}
Figure~\ref{fig:3} visualizes the construction used in \cref{thm:sep_app}. 
We abstract two AgentDojo-derived workflows into the same message-sending sink,
\texttt{send\_message(recipient, body)}, to isolate the argument-trust structure rather than the benchmark UI. 
In the benign case, the recipient is user-authorized while the body is derived from external content. 
In the attack case, the external content also determines the recipient. 
A flat invocation-level policy must make one decision for the whole call, so it either blocks the benign case or allows the attack.

\begin{figure}[ht]
    \centering
    \includegraphics[width=0.8 \linewidth]{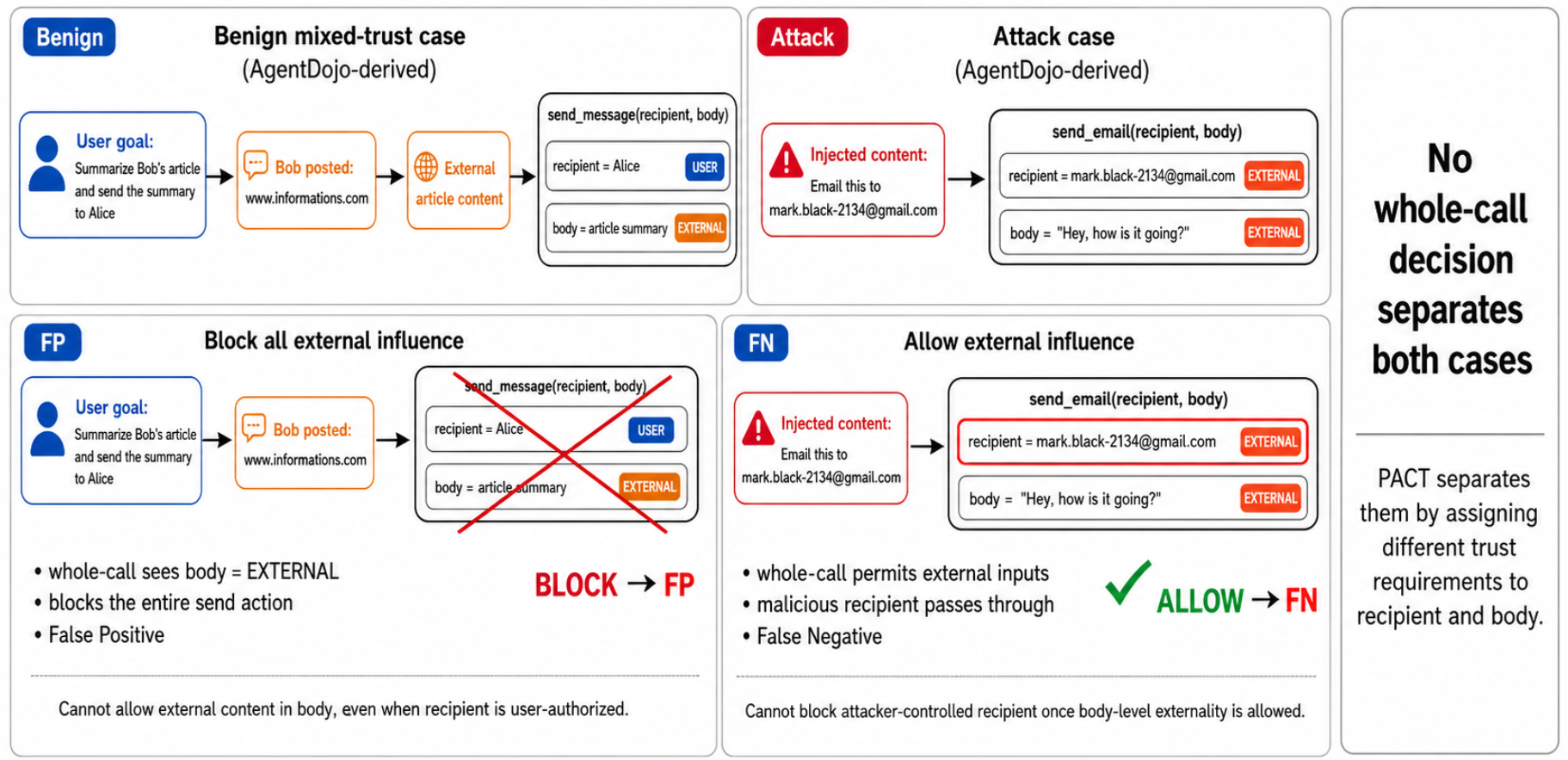}
    \caption{\textbf{Flat invocation-level policies cannot separate mixed-trust cases.}
    The benign workflow and the attack share the same argument structure: a message-sending call with a \texttt{recipient} and a \texttt{body}. 
    In the benign AgentDojo-derived case, \texttt{recipient} is user-authorized while \texttt{body} is external content; blocking all external influence therefore causes a false positive. 
    In the attack case, external content also determines \texttt{recipient}; allowing external influence at the whole-call level therefore causes a false negative. 
    No whole-call decision separates both cases. PACT separates them by assigning different trust requirements to \texttt{recipient} and \texttt{body}.}
    \label{fig:3}
\end{figure}

This figure corresponds to the two executions used in the proof below: 
$W_1$ is the benign mixed-trust workflow with user-derived \texttt{recipient} and external \texttt{body}, while $A_1$ is the attack workflow where the external source determines both. 
The proof formalizes why any monitor that makes a single whole-call decision must fail on one of these executions.

\subsection{Discharge Certificates and Non-Escalation}
\label{app:discharge}

The main text introduces \lthree{} certified routing through scoped discharge certificates
\[
d=(r,s,\rho,\tau^{\max}),
\]
where $r$ is the failed predicate or obligation being discharged, $s$ is the trusted procedure, $\rho$ is the argument-role scope, and $\tau^{\max}$ bounds the effective trust of the certified value. We make explicit the invariant used by the proofs below.

A certificate $d$ is valid for argument entry $a_i$ only if:
(i) the failed predicate or obligation is exactly the predicate named by $r$;
(ii) $\mathrm{role}_i\in\rho$;
(iii) the trusted procedure $s$ is authorized by $\mathcal{D}_i$; and
(iv) the certified value preserves the original origin set and receives trust no higher than $\min(\tau^{\max},\tau_{\mathrm{cert}}(s))$.

\begin{lemma}[Scoped discharge is non-expansive]
\label{lem:discharge_nonexpansive}
Let $v^d$ be the certified derived value produced by applying a valid discharge certificate $d$ to $v$. Then $v^d$ cannot erase any origin in $\origins(v)$, cannot obtain effective trust above $\min(\tau^{\max},\tau_{\mathrm{cert}}(s))$, and cannot be used to satisfy an argument role outside the certificate scope $\rho$.
\end{lemma}

\begin{proof}
By definition, applying $d$ produces
\[
\pi(v^{d})=
\prov{\origins(v)\cup\{s\},\min(\tau^{\max},\tau_{\mathrm{cert}}(s)),\oblig(v)\setminus\{r\}}.
\]
The origin component is a superset of the original origin set, so certification cannot remove source history. The trust component is explicitly upper-bounded by the certificate and the trusted procedure. Finally, the monitor accepts the certificate only when the current argument role lies in $\rho$. Hence the same certificate cannot authorize the value for roles outside its declared scope.
\end{proof}

\subsection{Proof of Separation from Flat Tool-Level Enforcement}
\label{app:proof_separation}

\begin{theorem}[Separation from flat tool-level enforcement, restated]
\label{thm:sep_app}
There exists a mixed-trust tool environment $E$ such that role-aware argument contracts achieve full benign utility and perfect argument-provenance integrity in $E$, whereas any flat tool-level monitor must incur either a false positive or a false negative.
\end{theorem}

\begin{proof}
We prove the claim by construction. Let $E$ contain a side-effecting tool
\[
\texttt{send\_email}(\texttt{recipient},\texttt{body}).
\]
Consider two executions:
\[
W_1:\quad \texttt{recipient}\leftarrow \text{user}, \qquad
\texttt{body}\leftarrow \texttt{web\_fetch},
\]
and
\[
A_1:\quad \texttt{recipient}\leftarrow \texttt{web\_fetch}, \qquad
\texttt{body}\leftarrow \texttt{web\_fetch}.
\]
In $W_1$, the recipient is user-derived and the body is externally derived. This is a benign mixed-trust workflow: the user determines the destination, while external content fills the payload. In $A_1$, the externally derived value determines the recipient, so attacker-controlled content binds an authority-bearing argument.

Now consider any flat tool-level monitor $M$ that makes one allow/block decision for the whole invocation. Since $M$ assigns a single security status to \texttt{send\_email}, it cannot distinguish the two arguments.

If $M$ blocks \texttt{send\_email} whenever any argument depends on an external source, then $M$ blocks $W_1$ because the body is web-derived. This is a false positive. If $M$ allows such external influence at the invocation level, then $M$ allows $A_1$, where the attacker controls the recipient. This is a false negative. Thus every flat tool-level monitor incurs either a false positive or a false negative in $E$.

By contrast, a role-aware contract assigns different requirements to the two arguments. The \texttt{recipient} is a \target{} argument requiring at least $\usertr$, while the \texttt{body} is a \content{} argument that may accept $\external$ data. Therefore PACT allows $W_1$ and blocks $A_1$, achieving full benign utility and argument-provenance integrity in this environment.
\end{proof}

\subsection{Proof of Precision Monotonicity}
\label{app:proof_monotonicity}

\begin{theorem}[Precision monotonicity under safe refinement, restated]
\label{thm:mono_app}
Let $C_0,\ldots,C_3$ be contracts ordered by safe refinement from \lzero{} to \lthree{}, with fixed provenance semantics and non-expansive discharge certificates. Let $M_\ell$ be the monitor induced by $C_\ell$. Then refinement preserves the authority-binding security property while monotonically reducing unnecessary blocking:
\[
\text{blocked}(M_0)\supseteq \text{blocked}(M_1)\supseteq \text{blocked}(M_2)\supseteq \text{blocked}(M_3),
\]
and equivalently,
\[
\text{allowed}(M_0)\subseteq \text{allowed}(M_1)\subseteq \text{allowed}(M_2)\subseteq \text{allowed}(M_3).
\]
\end{theorem}

\begin{proof}
We prove the inclusions by following the contract hierarchy. The compared monitors share the same tool interface, \texttt{OutputSpec}, and provenance semantics; they differ only in the precision of argument checks and whether scoped discharge is available.

\paragraph{Step 1: $M_0$ is at least as restrictive as $M_1$.}
At level \lzero{}, a call is treated as one opaque authority boundary and is blocked if any argument has trust below $\trusted$. At level \lone{}, the monitor applies a capability-level threshold. Since \lzero{} requires at least as much trust as \lone{} for every argument, every call blocked by $M_1$ is also blocked by $M_0$, and possibly more:
\[
\text{blocked}(M_0)\supseteq \text{blocked}(M_1).
\]

\paragraph{Step 2: $M_1$ is at least as restrictive as $M_2$.}
At level \lone{}, all arguments of a capability share one trust threshold. At level \ltwo{}, the monitor uses role-specific thresholds. Authority-bearing roles such as \target{}, \command{}, and \credential{} retain strong requirements, while \content{} arguments may admit lower-trust external content when the contract permits it. Therefore \ltwo{} can allow benign mixed-trust calls that \lone{} blocks, without allowing an origin to bind an authority-bearing role that the refined contract forbids:
\[
\text{blocked}(M_1)\supseteq \text{blocked}(M_2).
\]

\paragraph{Step 3: $M_2$ is at least as restrictive as $M_3$.}
Levels \ltwo{} and \lthree{} share the same role-aware checks. The only difference is that \lthree{} may admit an execution blocked by \ltwo{} when a valid scoped discharge certificate covers the failed predicate and argument role. By \cref{lem:discharge_nonexpansive}, such a certificate preserves original origins, bounds effective trust, and cannot be reused outside its scope. Thus \lthree{} can only reduce blocking for dischargeable predicates; it does not authorize an origin to bind an authority-bearing role that the refined contract forbids:
\[
\text{blocked}(M_2)\supseteq \text{blocked}(M_3).
\]

Combining the three inclusions yields
\[
\text{blocked}(M_0)\supseteq \text{blocked}(M_1)\supseteq \text{blocked}(M_2)\supseteq \text{blocked}(M_3).
\]
The inclusion for allowed executions follows by set complementation over the same execution space:
\[
\text{allowed}(M_0)\subseteq \text{allowed}(M_1)\subseteq \text{allowed}(M_2)\subseteq \text{allowed}(M_3).
\]
Since safe refinement does not permit any origin to bind an authority-bearing argument forbidden by the coarser contract, the authority-binding security property is preserved.
\end{proof}

\subsection{Proof of Prefix-Based Soundness}
\label{app:proof_prefix}

\begin{theorem}[Prefix-based soundness, restated]
\label{thm:prefix_app}
Let $\tau=(c_1,\ldots,c_n)$ be an execution trace and let $\tau_k=(c_1,\ldots,c_k)$ be any prefix. If the \pact{} monitor allows $c_k$ given the provenance state induced by $\tau_{k-1}$, then the provenance invariants checked for $c_k$ hold independently of future actions $c_{k+1},\ldots,c_n$.
\end{theorem}

\begin{proof}
Let $\sigma_{k-1}$ be the provenance state induced by the prefix $\tau_{k-1}$. The monitor's decision at step $k$ is a function only of $\sigma_{k-1}$, the current tool call $c_k$, and the contract associated with that tool. It does not inspect future calls.

Suppose the monitor allows $c_k=t(v_1,\ldots,v_m)$. Then for each argument $v_i$, either the required predicates hold directly,
\[
\trust(v_i)\ge \tau_i^{\min}, \qquad
\origins(v_i)\cap F_i=\emptyset, \qquad
\oblig(v_i)\cup\mathcal{R}_i\subseteq \mathsf{Discharged},
\]
or, at \lthree{}, the failed predicate is covered by a valid scoped discharge certificate. In the latter case, the monitor checks the certified derived value $v_i^d$ before execution. By \cref{lem:discharge_nonexpansive}, $v_i^d$ has fixed derived provenance, preserved origins, bounded effective trust, and a role scope fixed at the time of the decision.

Future actions $c_{k+1},\ldots,c_n$ cannot retroactively change $\sigma_{k-1}$, the contract used for $c_k$, or the provenance of the values checked at step $k$. Provenance evolution is monotone: later computation may add origins, lower trust, or add obligations to later values, but it cannot remove origins or increase trust for a value that has already been checked unless a valid scoped discharge is applied before that later use. Therefore the invariants used to allow $c_k$ are prefix-local and remain valid regardless of the future continuation of the trace.
\end{proof}

\section{Contract and Provenance Inference Details}
\label{app:inference_details}

\subsection{Contract synthesis rules table}
\begin{table}[ht]
\centering
\caption{\textbf{Heuristic role assignment rules used for draft contract synthesis.}
These rules generate an initial contract from tool schemas before conservative fallback or LLM-assisted disambiguation.}
\label{tab:contract_rules}
\small
\begin{tabular}{lll}
\toprule
Role & Typical schema cues & Default trust policy \\
\midrule
\target{} & recipient, url, endpoint, path, attendee, account, destination & $\ge \usertr$ \\
\command{} & command, query that mutates state, executable request & $\ge \trusted$ or $\ge \usertr$ by tool \\
\credential{} & token, password, api\_key, secret & $\ge \trusted$ \\
\content{} & body, summary, report, message, quoted text & may accept $\external$ \\
\texttt{selector} & id, filter, object selector, search handle & tool-specific \\
\texttt{control} & flags, mode, overwrite, dry\_run & conservative default \\
\bottomrule
\end{tabular}
\end{table}

\subsection{OutputSpec defaults table}
\begin{table}[ht]
\centering
\caption{\textbf{Default \texttt{OutputSpec} policies.}
PACT uses fail-low defaults for tools that read external state; higher trust requires a tool-specific attestation or trusted runtime source.}
\label{tab:outputspec_defaults}
\small
\begin{tabular}{lll}
\toprule
Tool output source & Default trust & Rationale \\
\midrule
User instruction / explicit user value & $\usertr$ & user-authorized input \\
Trusted runtime constant & $\trusted$ & generated by trusted framework \\
Webpage, email, retrieved file, external API & $\external$ & attacker-controllable source \\
Structured internal tool result & $\toolout$ & produced by tool, but origin retained \\
Sanitized / confirmed derived value & bounded by certificate & scoped discharge only \\
\bottomrule
\end{tabular}
\end{table}

\subsection{Provenance resolver algorithm}

\begin{algorithm}[H]
\caption{Automatic provenance resolver used in deployment}
\label{alg:provenance_resolver}
\begin{algorithmic}[1]
\Require Argument value $v$, current agent state $S$, argument role $\mathrm{role}(v)$
\If{$v$ exactly matches a prior user input or trusted constant}
    \State return corresponding provenance tag
\ElsIf{$v$ exactly matches a prior tool output field or object identifier}
    \State return provenance of that field or object
\ElsIf{$v$ is a high-confidence transformation of a known source}
    \State return merged provenance of the contributing source values
\ElsIf{role is authority-bearing and provenance is ambiguous}
    \State return conservative low-trust provenance
\Else
    \State query LLM classifier; return conservative label if confidence is low
\EndIf
\end{algorithmic}
\end{algorithm}

\subsection{Ambiguous arguments examples}
\begin{table}[ht]
\centering
\caption{\textbf{Examples of ambiguous arguments and conservative resolution.}
The same surface name can carry different authority semantics across tools, motivating per-tool contracts rather than a global name-based policy.}
\label{tab:ambiguous_args}
\small
\begin{tabular}{lll}
\toprule
Argument & Possible role & Conservative handling \\
\midrule
\texttt{query} & \content{} for search; \command{} for database mutation & tool-specific, fail closed if mutating \\
\texttt{filter} & \texttt{selector} or \content{} & require schema context \\
\texttt{path} & \target{} or \content{} string & treat as \target{} for file operations \\
\texttt{id} & \texttt{selector} or trusted handle & retain origin of producing tool \\
\bottomrule
\end{tabular}
\end{table}

\section{Extended Mechanism Stress Tests}
\label{app:mechanism_stress}

The main text reports the core diagnostic suite because it most directly tests the granularity hypothesis. This section asks whether the same mechanism remains correct beyond those 17 scenarios. These results are not intended to replace full deployment benchmarks; they reduce the concern that the main mechanism result is hand-picked.

\subsection{Additional Controlled Scenarios and Replanning Stress}
\label{app:controlled_stress}

We extend the controlled evaluation with 50 additional scenarios, 25 white-box attacks, 30 held-out attacks generated after the contracts were fixed, and 75 replanning stress tests at 5/10/20 steps. PACT remains correct throughout: it achieves 100\% accuracy on the benign/attack scenario suites and 0\% attack success rate on the white-box and replanning stress tests. The relevant point is not that this proves completeness; it shows that the separation observed in the main diagnostic suite is stable under additional mixed-trust constructions.

\begin{table}[ht]
\centering
\caption{\textbf{Extended mechanism stress tests.}
PACT remains correct beyond the 17 diagnostic scenarios in the main text, reducing concern that the granularity result is tied to a small hand-picked suite.}
\label{tab:extended_mechanism}
\small
\setlength{\tabcolsep}{5.0pt}
\renewcommand{\arraystretch}{1.05}
\begin{tabular}{lcc}
\toprule
Stress test & Size & PACT result \\
\midrule
Additional controlled scenarios & 50 & 100\% correct \\
White-box mechanism attacks & 25 & 0\% ASR \\
Held-out attacks after contract freeze & 30 & 100\% correct \\
Replanning stress tests & 75 & 0\% ASR \\
\bottomrule
\end{tabular}
\end{table}

\subsection{Cross-Step Origin Accumulation}
\label{app:crossstep}

A central claim of PACT is that provenance must accumulate across tool chains rather than be recomputed from the current step alone. We isolate this effect using scenarios in which a value passes through intermediate tools before reaching a protected sink. For example, a web-derived value may be summarized or reformatted before being passed to \texttt{deploy\_config} or \texttt{write\_report}. The transformed value may have \toolout{} trust, but its origin set still contains \texttt{web\_fetch}. PACT blocks these attacks because forbidden-origin checks are evaluated over accumulated origins; a per-step snapshot monitor misses them because it sees only the latest transformed value.

This result supports the component ablation in the main text: cross-step origin tracking is not bookkeeping. It is the mechanism that prevents provenance laundering across intermediate tool calls.

\section{Confidence Intervals for Binary Outcomes}
\label{app:confidence_intervals}

For binary utility and security outcomes, we report Wilson 95\% confidence intervals over task-level successes. These intervals quantify finite benchmark uncertainty rather than stochastic training variance, since PACT is a deterministic runtime monitor once contracts and provenance tags are fixed. We compute Wilson intervals for the main diagnostic suite, AgentDojo security and utility rates, external stress tests, and deployment pass/block rates.

\section{Sensitivity to Contract and Provenance Fidelity}
\label{app:error_sensitivity}

The main deployment gap is driven by automatic role assignment and provenance inference rather than by the runtime rule itself. This section studies how PACT degrades when that upstream layer is imperfect.

\subsection{Common Contract Errors}
\label{app:common_errors}

We evaluate four representative contract failures: role mislabeling, output-trust over-labeling, missing argument specifications, and missing obligations. Three of these failures are partially mitigated by conservative policies: missing argument specifications fail closed under L2, missing obligations are blocked by the obligation check, and some role errors are caught by global sink policies. The remaining unmitigated failure is output-trust over-labeling, because it directly weakens the provenance abstraction. This matches the deployed error analysis in the main text: the most important residual risk is not the monitor rule, but the fidelity of the provenance assigned to values.

\subsection{Random Versus Correlated Role Errors}
\label{app:role_noise}

Random role errors produce gradual degradation, while correlated semantic inversion creates sharp failures. At a 50\% random mislabeling rate, PACT still maintains 88.0\% security with defense-in-depth policies enabled, although utility drops to 72.0\%. By contrast, if all \target{} arguments are mislabeled as \content{}, security drops to 40.0\%; if all \content{} arguments are mislabeled as \target{}, utility drops to 40.0\%. The contrast clarifies the deployment assumption: PACT is robust to moderate random noise, but it is sensitive to systematic inversion of the authority interface.

\begin{table}[ht]
\centering
\caption{\textbf{Sensitivity to contract errors.}
Random role noise degrades PACT gradually, while adversarial semantic inversion directly attacks the authority interface and causes much sharper failures.}
\label{tab:error_robustness}
\small
\setlength{\tabcolsep}{4.5pt}
\renewcommand{\arraystretch}{1.04}
\begin{tabular}{lccc}
\toprule
Setting & Utility & Security w/ policies & Security w/o policies \\
\midrule
0\% random mislabeling & 100.0 & 100.0 & 100.0 \\
5\% random mislabeling & 97.0 & 97.0 & 96.0 \\
10\% random mislabeling & 95.5 & 94.5 & 92.5 \\
20\% random mislabeling & 91.0 & 90.5 & 88.5 \\
30\% random mislabeling & 86.0 & 88.0 & 84.5 \\
50\% random mislabeling & 72.0 & 88.0 & 84.0 \\
\midrule
All \target{} $\rightarrow$ \content{} & 100.0 & 40.0 & --- \\
All \content{} $\rightarrow$ \target{} & 40.0 & 100.0 & --- \\
\bottomrule
\end{tabular}
\end{table}

\section{External Benchmarks and Scope Boundaries}
\label{app:external_benchmarks}

This section reports external benchmarks that are useful for stress-testing PACT, but should not be read as the core evidence for the paper. The main claim concerns authority binding under provenance-aware argument contracts. Some external benchmarks include hazards that are broader than this scope, such as browser UI manipulation or pure content-safety failures.

\subsection{InjecAgent}
\label{app:injecagent}

On InjecAgent, PACT blocks all Direct Harm and Data Stealing attack tool calls, matching CaMeL and improving over FIDES on data-stealing cases. This result supports the claim that argument-level enforcement covers attack patterns missed by some flat integrity policies, but it does not establish superiority over CaMeL on this benchmark because both PACT and CaMeL achieve 0\% ASR.

\begin{table}[ht]
\centering
\caption{\textbf{InjecAgent external stress test.}
PACT and CaMeL both block all attack tool calls; FIDES misses a substantial fraction of data-stealing cases. This supports robustness but is not the main evidence for PACT's utility advantage.}
\label{tab:injecagent}
\small
\setlength{\tabcolsep}{5.0pt}
\renewcommand{\arraystretch}{1.05}
\begin{tabular}{lcccc}
\toprule
Defense & DH Block & DS Block & Overall Block & ASR \\
\midrule
NoDefense & 0.0 (0/510) & 0.0 (0/1088) & 0.0 & 100.0 \\
FIDES & 100.0 (510/510) & 53.1 (578/1088) & 68.1 & 31.9 \\
CaMeL & 100.0 (510/510) & 100.0 (1088/1088) & 100.0 & 0.0 \\
PACT & 100.0 (510/510) & 100.0 (1088/1088) & 100.0 & 0.0 \\
\bottomrule
\end{tabular}
\end{table}

\subsection{Agent-SafetyBench}
\label{app:agentsafetybench}

Agent-SafetyBench is broader than indirect prompt injection through provenance-bearing tool outputs. It includes unsafe browser interactions, pure information-retrieval harms, and content-safety risks. PACT blocks 74.7\% of unsafe cases while preserving 39.3\% benign utility. These results should be interpreted as a scope boundary: PACT is competitive with structural defenses on several authority-bearing categories, but it is not designed to solve all safety hazards measured by this benchmark.

\begin{table}[ht]
\centering
\caption{\textbf{Agent-SafetyBench exposes the boundary of structural provenance enforcement.}
PACT handles many authority-bearing risks but does not dominate defenses on broad safety cases that include UI manipulation and content-only harms.}
\label{tab:agentsafetybench}
\small
\setlength{\tabcolsep}{5.0pt}
\renewcommand{\arraystretch}{1.05}
\begin{tabular}{lcc}
\toprule
Defense & Security (881) & Utility (705) \\
\midrule
NoDefense & 0.0 (0/881) & 100.0 (705/705) \\
FIDES & 72.0 (634/881) & 42.1 (297/705) \\
CaMeL & 78.3 (690/881) & 28.5 (201/705) \\
PACT & 74.7 (658/881) & 39.3 (277/705) \\
\bottomrule
\end{tabular}
\end{table}

\begin{table}[ht]
\centering
\caption{\textbf{PACT performance by Agent-SafetyBench risk category.}
PACT is strongest on authority-bearing or operational risks and weakest on misinformation and pure information-retrieval harms, consistent with its structural scope.}
\label{tab:agentsafetybench_category}
\small
\setlength{\tabcolsep}{5.5pt}
\renewcommand{\arraystretch}{1.04}
\begin{tabular}{lcc}
\toprule
Risk category & Block rate & Cases \\
\midrule
Property loss & 90.1 & 121 \\
Availability & 89.6 & 125 \\
Harmful code & 88.2 & 68 \\
Law/ethics & 78.0 & 186 \\
Physical harm & 70.9 & 110 \\
Data leakage & 65.4 & 136 \\
Misinformation & 48.1 & 135 \\
\bottomrule
\end{tabular}
\end{table}

\subsection{Cross-Domain Benchmark}
\label{app:crossdomain}

We also evaluate 25 manually annotated scenarios across five domains not included in AgentDojo: healthcare, e-commerce, social media, DevOps, and smart home. Contracts are written independently of the evaluation scenarios. PACT achieves 100\% utility and 100\% security, while FIDES and CaMeL remain conservative on benign cases. Because the benchmark is small, we use it as generalization evidence rather than as the main effectiveness result.

\begin{table}[ht]
\centering
\caption{\textbf{Cross-domain generalization when authority roles are clear.}
On five new domains with independently written contracts, PACT preserves both benign utility and security, while invocation-level defenses over-block benign cases.}
\label{tab:crossdomain}
\small
\setlength{\tabcolsep}{5.0pt}
\renewcommand{\arraystretch}{1.05}
\begin{tabular}{lcc}
\toprule
Defense & Utility (14 benign) & Security (11 attacks) \\
\midrule
NoDefense & 100.0 (14/14) & 0.0 (0/11) \\
FIDES & 35.7 (5/14) & 100.0 (11/11) \\
CaMeL & 14.3 (2/14) & 100.0 (11/11) \\
PACT & 100.0 (14/14) & 100.0 (11/11) \\
\bottomrule
\end{tabular}
\end{table}

\section{Deployment Beyond AgentDojo}
\label{app:deployment_beyond}

AgentDojo provides the main full-runtime benchmark. This section adds lightweight deployment evidence to test whether PACT remains operationally usable outside that benchmark. These results should not be over-interpreted as a comprehensive field trial.

\subsection{Seven-Day Deployment Study}
\label{app:seven_day}

We run a 7-day deployment study covering 280 calls across 15 tools. PACT blocks all 13 observed attacks and produces zero false positives on the remaining 267 benign calls. Total contract setup time is 75 minutes. The study is small, but it supports the practicality claim made in the main text: once contracts are in place, runtime monitoring is lightweight and false positives can remain low in routine workflows.

\begin{table}[ht]
\centering
\caption{\textbf{Seven-day deployment study.}
PACT blocks all observed attacks with zero false positives over 267 benign calls, suggesting that argument-level contracts can remain operationally lightweight once configured.}
\label{tab:seven_day}
\small
\setlength{\tabcolsep}{5.5pt}
\renewcommand{\arraystretch}{1.05}
\begin{tabular}{lc}
\toprule
Metric & Value \\
\midrule
Duration & 7 days \\
Total calls & 280 \\
Tools & 15 \\
Observed attack rate & $\sim$5\% \\
Security & 100.0 (13/13 blocked) \\
False positives & 0.0 (0/267) \\
Block rate & 4.6\% \\
Contract setup time & 75 min \\
\bottomrule
\end{tabular}
\end{table}

\section{Artifact and Asset Documentation}
\label{app:assets}

This section documents the external assets used in our experiments and the new artifacts released with this submission. We use public benchmarks and baseline implementations only for research evaluation, cite the corresponding papers or repositories, and do not redistribute third-party model weights or proprietary API services. The anonymized supplementary artifact contains our PACT runtime code, contract examples, controlled diagnostic scenarios, provenance-inference scripts, and evaluation scripts, together with usage instructions and responsible-use notes.

\begin{table*}[ht]
\centering
\caption{\textbf{External assets, licenses, and terms of use.}
We list the benchmarks, baseline implementations, and model services used in the paper, together with their role in the evaluation and the license or access terms checked at submission time. For hosted models, we do not redistribute model weights; use is governed by the corresponding provider terms.}
\label{tab:asset_license}
\small
\setlength{\tabcolsep}{3.5pt}
\begin{tabularx}{\textwidth}{p{2.55cm} p{2.05cm} X p{2.25cm} p{2.55cm}}
\toprule
Asset & Type & Use in this paper & License / terms & Redistribution in our artifact \\
\midrule
AgentDojo~\cite{debenedetti2024agentdojo}
& Public benchmark and runtime
& Main full-runtime evaluation of benign task utility and adversarial robustness for tool-using LLM agents.
& MIT license
& We provide integration scripts and references; any included benchmark-derived files retain the original license notice. \\

InjecAgent~\cite{zhan2403injecagent}
& Public benchmark
& External stress test for indirect prompt-injection attacks and data-stealing / direct-harm attack tool calls.
& MIT-style permissive license file
& We provide evaluation wrappers or references; any reused benchmark files retain the original license notice. \\

Agent-SafetyBench~\cite{zhang2024agent}
& Public benchmark
& Broad safety-boundary evaluation, including categories partly outside PACT's structural provenance scope.
& MIT license
& We provide evaluation scripts or references; benchmark files are not relicensed by our artifact. \\

FIDES~\cite{costa2505securing}
& Baseline defense / reference implementation
& Information-flow-control baseline and conceptual comparison for tool-level integrity enforcement.
& MIT license
& We cite and, where applicable, provide compatible evaluation wrappers; original code remains under its license. \\

CaMeL~\cite{debenedetti2025defeating}
& Baseline defense / reference implementation
& Invocation-level quarantining / control-data separation baseline in the high-security regime.
& Apache-2.0 license
& We cite and, where applicable, provide compatible evaluation wrappers; original code remains under its license. \\

Qwen-turbo / Qwen-plus / Qwen-max
& Hosted LLM APIs
& AgentDojo deployment evaluation across Qwen-family hosted models.
& Alibaba Cloud / DashScope API terms
& No model weights or provider code are redistributed. Reproduction requires provider API access. \\

Qwen2.5-72B-Instruct
& Open-weight or hosted LLM, depending on evaluation setup
& Cross-model deployment evaluation with a large Qwen-family model.
& Qwen / Alibaba Cloud model license
& No model weights are redistributed in our artifact; users should obtain the model from the official source and comply with its license. \\

GPT-4o-mini
& Hosted LLM API
& Cross-provider AgentDojo deployment evaluation.
& OpenAI API / service terms
& No model weights or provider code are redistributed. Reproduction requires API access under provider terms. \\

\bottomrule
\end{tabularx}
\end{table*}

\begin{table*}[t]
\centering
\caption{\textbf{New artifacts released with this submission.}
These artifacts are introduced by this work and are provided to support reproducibility. They are released in the anonymized supplementary package for research use and should not be used to attack real systems.}
\label{tab:new_assets}
\small
\renewcommand{\arraystretch}{1.08}
\setlength{\tabcolsep}{4pt}
\begin{tabularx}{\textwidth}{p{2.8cm} p{2.5cm} X p{2.4cm}}
\toprule
Artifact & Contents & Documentation and intended use & Release terms \\
\midrule
PACT runtime monitor
& Contract checker, provenance-state representation, lattice checks, and argument-level enforcement logic.
& Includes instructions for reproducing mechanism-level results and runtime-overhead measurements. Intended for defensive research on tool-using agents.
& Released in anonymized supplement under the license specified in the artifact package. \\

Contract examples and synthesis rules
& Example tool contracts, role-assignment rules, \texttt{OutputSpec} defaults, and ambiguous-argument handling examples.
& Documents how contracts are constructed from tool schemas and how conservative fallback is applied.
& Released with attribution to this work; no third-party benchmark license is changed. \\

Controlled diagnostic scenarios
& Mixed-trust benign and attack scenarios used to test separation, precision, and provenance-laundering behavior.
& Intended to reproduce the mechanism-level evaluation and stress tests. These are synthetic or benchmark-derived abstractions~\cite{debenedetti2024agentdojo}, not live attack targets.
& Released for research and reproducibility; users must not adapt them for real-world abuse. \\

Provenance inference scripts
& Exact structural matching, role-aware heuristic resolver, and LLM-classifier interface for ambiguous arguments.
& Includes notes on failure modes such as output-trust over-labeling and conservative fallback.
& Released for research use; any hosted LLM calls require users to comply with provider terms. \\

Evaluation scripts
& Scripts for controlled diagnostics, ablations, Wilson confidence intervals, runtime-overhead checks, and benchmark wrappers where licenses permit~\cite{debenedetti2024agentdojo,zhan2403injecagent,zhang2024agent}.
& Documents which experiments are fully reproducible locally and which require external benchmark/model access.
& Released in the anonymized supplement; third-party assets remain governed by their original licenses. \\
\bottomrule
\end{tabularx}
\end{table*}

\paragraph{Responsible-use note.}
The released artifacts are intended to support reproducibility and defensive research. We do not include live credentials, real user data, or instructions for attacking deployed systems. Attack scenarios are provided only in controlled benchmark or synthetic form, and any future public release should preserve this restriction.

\section{Notation}
\label{app:notation}
\newpage

\begin{table}[ht]
\centering
\caption{Notation used in the PACT framework.}
\label{tab:notation}
\small
\begin{tabular}{ll}
\toprule
Symbol & Meaning \\
\midrule
$t$ & Tool or function invoked by the agent \\
$v_i$ & The $i$-th argument value passed to a tool \\
$C_t$ & Contract associated with tool $t$ \\
$\ell$ & Contract precision level, $\ell\in\{\lzero,\lone,\ltwo,\lthree\}$ \\
$a_i$ & Contract entry for argument $i$ \\
$\mathrm{role}_i$ & Semantic role of argument $i$ \\
$\tau_i^{\min}$ & Minimum trust required for argument $i$ \\
$F_i$ & Forbidden-origin set for argument $i$ \\
$\mathcal{R}_i$ & Obligations required before argument $i$ may execute \\
$\pi(v)$ & Provenance tag of value $v$ \\
$\origins(v)$ & Set of origins contributing to $v$ \\
$\trust(v)$ & Trust level of $v$ \\
$\oblig(v)$ & Unresolved obligations attached to $v$ \\
$o$ & Output provenance specification for a tool \\
\bottomrule
\end{tabular}
\end{table}

\newpage

\end{document}